\documentclass[draftclsnofoot,onecolumn]{IEEEtran}

\usepackage[utf8]{inputenc}
\usepackage{url}
\usepackage{cite}
\usepackage[cmex10]{amsmath}
\usepackage{amssymb}
\usepackage{amsthm}
\usepackage{xcolor}

\DeclareMathOperator*{\argmin}{arg\,min}       
\DeclareMathOperator{\I}{\mathbb{I}}           
\newcommand{\set}[1]{\mathcal{#1}}             
\DeclareMathOperator{\tclOp}{\set{T}}
\newcommand{\tcl}[1]{\tclOp^{(#1)}}            
\DeclareMathOperator{\typOp}{\set{P}}
\newcommand{\typ}[1]{\typOp^{(#1)}}            
\DeclareMathOperator{\Prb}{Pr}                 
\DeclareMathOperator{\E}{E}                    
\DeclareMathOperator{\Entr}{H}                 
\DeclareMathOperator{\Info}{I}                 
\DeclareMathOperator{\RDOp}{R}
\newcommand\RD[2]{\RDOp_{#1, #2}}              
\DeclareMathOperator{\dist}{d}                 
\DeclareMathOperator{\D}{D}                    
\DeclareMathOperator{\IID}{IID}                
\newcommand{\pmf}{\mathcal{P}}

\newtheorem{theorem}{Theorem}
\newtheorem{corollary}{Corollary}
\newtheorem{lemma}{Lemma}
\newtheorem{remark}{Remark}

\ifCLASSINFOpdf
\else
\fi
\begin{document}

\title{Guessing Based on Compressed\\ Side Information}

\author{
  \IEEEauthorblockN{Robert Graczyk\IEEEauthorrefmark{1},
    Amos Lapidoth\IEEEauthorrefmark{1},
    Neri Merhav\IEEEauthorrefmark{2}},
    and Christoph Pfister\IEEEauthorrefmark{1}\\
  \IEEEauthorblockA{
    \IEEEauthorrefmark{1}Signal and Information Processing Laboratory ETH Zurich,
    8092 Zurich, Switzerland\\
    Email: \{graczyk, lapidoth, pfister\}@isi.ee.ethz.ch\\
    \IEEEauthorrefmark{2}The Viterbi Faculty of Electrical and Computer Engineering Technion,\\
    Technion City, Haifa 3200003, Israel\\
    Email: merhav@ee.technion.ac.il
  }
}


\maketitle

\begin{abstract}
  A source sequence is to be guessed with some fidelity based on a
  rate-limited description of an observed sequence with which it is
  correlated. 
  The trade-off between the description rate and the exponential
  growth rate of the least power mean of the number of guesses is
  characterized.
\end{abstract}

\begin{IEEEkeywords}
  Compression, Guessing, Side Information.
\end{IEEEkeywords}

%
\IEEEpeerreviewmaketitle

\section{Introduction}
Our problem can be viewed as the guessing analogue of the Remote Sensing
problem in lossy source coding \cite{DobTsy62}, \cite{Witsenhausen80},
\cite{WolfZiv70}. As in that problem, the description of a source
sequence is indirect: the rate-limited description is based only on a
noisy version of the sequence. The problems differ, however, in their
objectives: in the Remote Sensing problem the source sequence is
\emph{estimated} (with the least expected distortion), whereas in our
problem it is \emph{guessed} to within some distortion (with the least
power mean of the number of required guesses). Our problem thus
relates to Ar{\i}kan and Merhav's guessing-subject-to-distortion problem
\cite{ArikanMerhav98} in much the same way that the Remote Sensing
problem relates to Shannon's lossy source coding problem
\cite{Shannon59}.


To put our problem in context, recall that in the guessing problem
pioneered by Massey \cite{Massey94} and Ar{\i}kan \cite{Arikan96}, a
guesser seeks to recover a finite-valued chance variable $X \in \set{X}$
by sequentially producing guesses of the form
\begin{align*}
  &\textnormal{``Is $X = x_1?$''}\\
  &\textnormal{``Is $X = x_2?$''}\\
  &\phantom{\textnormal{``Is $X$ }}\,\vdots
\end{align*}
where $x_1, x_2, \ldots \in \set{X}$, and each guess is answered
truthfully with ``Yes'' or ``No.'' The number of guesses taken until
the first ``Yes,'' i.e., until $X$ is revealed, depends on the
guesser's strategy $\set{G}$ (the order in which the elements of
$\set{X}$ are guessed) and is denoted $G(X)$. Given the probability
mass function (PMF) $P_X$ of $X$ and some $\rho > 0$, Ar{\i}kan showed
\cite{Arikan96} that the least achievable $\rho$-th moment of the
number of guesses $\E[G(X)^\rho]$ required to recover $X$ is
closely related to its R\'{e}nyi entropy:
\begin{align}\label{eq:arikan_bounds}
  &\frac{1}{(1\! +\! \log\!|\set{X}|)^\rho}2^{\rho\Entr_{1\!/\!(1 + \rho)}(P_X)}\! \leq\! \min_{\set{G}} \E[G(X)^\rho]\! \leq\! 2^{\rho\Entr_{1\!/\!(1 + \rho)}(X)},
\end{align}
where $\Entr_{1/(1 + \rho)}(P_X)$ denotes the order-$1/(1 + \rho)$
R\'{e}nyi entropy of $X$. When guessing a length-$n$ random sequence
$X^n \triangleq (X_1, \ldots, X_n)$ whose components are independent
and identically distributed (IID) according to $P_X$,
Inequality~\eqref{eq:arikan_bounds} implies that
\begin{equation}
  \lim_{n \to \infty} \frac{1}{n} \log\Big(\!\min_{\set{G}} \E[G(X^n)^\rho]\Big) = \rho\Entr_{1/(1 + \rho)}(P_X),
\end{equation}
so the R\'{e}nyi entropy of $X$ fully characterizes (up to the factor
$\rho$) the exponential growth rate of the least $\rho$-th moment of the
number of guesses required to recover $X^n$.

Our problem differs from Massey's and Ar{\i}kan's in the following two ways:
\begin{enumerate}
  \item Instead of recovering $X^n$, the guesser need only produce a guess
  $\hat{X}^n \in \hat{\set{X}}^n$ that is close to $X^n$ in the sense that
  \begin{equation}\label{eq:dist_constraint}
    \frac{1}{n} \sum_{i = 1}^n d(X_i, \hat{X}_i) \leq D,
  \end{equation}
  where the distortion measure $d(\cdot, \cdot)\colon \set{X} \times
  \hat{\set{X}} \to \mathbb{R}_{\geq 0}$ and the maximal-allowed
  distortion level $D > 0$ are prespecified.
  We assume that, for every $x^n \in \set{X}^n$,
  \eqref{eq:dist_constraint} is satisfied by some
  $\hat{x}^{n} \in \hat{\set{X}}^n$; this guarantees the existence of
  a guessing strategy that eventually succeeds.

  
\item Prior to guessing, the guesser is provided with a rate-limited
  description $f(Y^{n}) \in \{0,1\}^{nR}$ of a noisy observation
  $Y^n \in \set{Y}^n$ of $X^{n}$. Based on $f(Y^{n})$, the guesser
  sequentially guesses elements $\hat{X}^n$ of $\hat{\set{X}}^n$ until
  \eqref{eq:dist_constraint} is satisfied.
  (The guesser's strategy $\set{G}$ thus depends on $f(Y^n)$.)
\end{enumerate}
We show that when $(X_1, Y_1), \ldots, (X_n, Y_n)$ are IID according
to $P_{XY}$, the exponential growth rate of the least $\rho$-th moment
of the number of guesses---optimized over the description function $f$
and the guessing strategy $\set{G}$---satisfies the variational
characterization \eqref{eq:target_exponent} of Theorem \ref{thm:main}
ahead.

Along the lines of \cite{Cachin97}, this theorem can be used to assess
the resilience of a password $X^{n}$ against an adversary who has
access to $nR$ bits of a correlated password $Y^{n}$ and is content
with guessing only a fraction $1 - D$ of the symbols of $X^{n}$. (In
this application, the distortion function is the Hamming distance.)



Since our guessing problem is an extension of the
guessing-subject-to-distortion problem studied by Merhav and Ar{\i}kan
\cite{ArikanMerhav98}, their suggested motivation (accounting for the
computational complexity of a rate-distortion encoder as measured by the
number of metric calculations) and proposed applications (betting games,
pattern matching, and database search algorithms) also extend to our setup.
Further applications include  sequential decoding \cite{Arikan96}, compression
\cite{Sundaresan07a}, and task encoding \cite{BraLapPfi17, BraLapPfi19}.

Numerous other variations on the Massey-Ar{\i}kan guessing problem were
studied over the years. Sundaresan \cite{Sundaresan07b} derived an
expression for the smallest guessing moment when the source
distribution is only partially known to the guesser; in
\cite{MerhavCohen20, SalBeiCohMed17}, the authors constructed and
analyzed optimal decentralized guessing strategies (for multiple
guessers that cannot communicate); in \cite{WeiSha18}, Weinberger and
Shayevitz quantified the value of a single bit of side-information
provided to the guesser prior to guessing; in \cite{BeiCalChrDufMed18},
the authors studied the guessing problem using an information-geometric
approach; and in \cite{GraLap20} and \cite{BraLapPfi19} the authors
studied the distributed guessing problem on Gray-Wyner and
Stelpian-Wolf networks.

The above distributed settings dealt, however, only with ``lossless''
guessing, where the guessing has to be exact. Our present setting
maintains, to some degree, a distributed flavor, but allows for
``lossy'' guessing, i.e., with some fidelity.

\section{Problem Statement and Notation}

Consider $n$ pairs $\{(X_i, Y_i)\}_{i = 1}^n$ that are drawn
independently, each according to a given PMF $P_{XY}$ on the finite
Cartesian product $\set{X} \times \set{Y}$:
\begin{equation}\label{eq:iid_source}
  \{(X_i, Y_i)\}_{i = 1}^n \sim \IID P_{XY}.
\end{equation}
Define the sequences
\begin{equation}
  X^n \triangleq \{X_i\}_{i = 1}^n,\: Y^n
\triangleq \{Y_i\}_{i = 1}^n,
\end{equation}
with $\{X_i\}_{i = 1}^n$ being $\IID P_X$, where $P_X$ is the
$X$-marginal of $P_{XY}$, and likewise $\{Y_i\}_{i = 1}^n$ being
$\IID P_Y$. By possibly redefining $\set{X}$ and $\set{Y}$, we assume
without loss of generality that $P_{X}$ and $P_{Y}$ are positive.
A guesser wishes to produce a sequence $\hat{X}^n$, taking
values in a finite $n$-fold Cartesian product set $\hat{\set{X}}^n$, that is ``close'' to $X^n$ in the
sense that
\begin{equation}\label{eq:thres_to_sat}
  \bar{\dist}(X^n, \hat{X}^n) \leq D,
\end{equation}
where $D > 0$ is some prespecified maximally-allowed
distortion, and 
\begin{equation}
  \bar{\dist}(x^n, \hat{x}^n) \triangleq \frac{1}{n}\sum_{i = 1}^n
  \dist(x_i, \hat{x}_i)
\end{equation}
with
\begin{equation}
  \dist\colon \set{X} \times \hat{\set{X}} \to \mathbb{R}_{\geq 0}
\end{equation}
some prespecified distortion function. We assume that $d(\cdot, \cdot)$
and $D$ are such that for each $x^n \in \set{X}^n$ there exists some 
$\hat{x}^n \in \hat{\set{X}}^n$ for which \eqref{eq:thres_to_sat} is
satisfied,
\begin{equation}\label{eq:valid_x_hat_assumption}
  \forall x^n \in \set{X}^n\: \exists \hat{x}^n \in \hat{\set{X}}^n\colon
  \bar{\dist}(x^n, \hat{x}^n) \leq D.
\end{equation}
This guarantees that such $\hat{X}^{n}$ can be found and in no-more-than 
$|\hat{\set{X}}|^n$ guesses.

Courtesy of a ``helper'' $f_n\colon \set{Y}^n \to \{0, 1\}^{nR}$, the guesser
is provided, prior to guessing, with an $nR$-bit description $f_{n}(Y^{n})$
of $Y^{n}$. Based on $f_n(Y^n)$, the guesser produces a ``guessing
strategy'' (also called a ``guessing function'')
\begin{equation}
  \set{G}_n\bigl(\,\cdot\, | f_{n}(Y^n)\bigr)\colon \{1, \ldots, |\hat{\set{X}}^n|\} \to \hat{\set{X}}^n,
\end{equation}
with the understanding that its first guess is
$\set{G}_{n}\bigl(1 \big| f_{n}(Y^n)\bigr)$, followed by
$\set{G}_{n}\bigl(2|f_{n}(Y^n)\bigr)$, etc. Thus, the guesser first asks
\begin{equation*}
  \text{``Does } \set{G}_{n}\bigl(1 \big| f_{n}(Y^n)\bigr) \text{ satisfy
  \eqref{eq:thres_to_sat}?''}
\end{equation*}
If the answer is ``yes,'' the guessing terminates and
$\set{G}_{n}\bigl(1 \big| f_{n}(Y^n)\bigr) \in \hat{\set{X}}^{n}$ is
produced. Otherwise the guesser asks
\begin{equation*}
  \text{``Does } \set{G}_{n}\bigl(2 \big| f_{n}(Y^n)\bigr) \text{ satisfy
    \eqref{eq:thres_to_sat}?''}
\end{equation*}
etc. Since guessing the same sequence twice is pointless, we assume
(without loss of optimality) that, for every value of $f_{n}(y^{n})$,
the mapping $\set{G}_n(\,\cdot\,|f_{n}(y^{n}))$ is injective and hence---since its
domain and codomain are of equal cardinality---bijective.
This and Assumption~\eqref{eq:valid_x_hat_assumption},
allow us to define
\begin{equation}\label{eq:def_guessing_fn}
  G_n\bigl(x^n \big| f_{n}(y^{n})\bigr) \triangleq \min\big\{i \geq 1\colon \bar{\dist}(x^n,
  \set{G}_n\bigl(i \big| f_{n}(y^{n})\bigr) \leq D\big\}
\end{equation}
as the number of required guesses when $X^{n} = x^{n}$ and
$f_{n}(Y^{n}) = f_{n}(y^{n})$.

Given a positive constant $\rho$, we seek the least exponential
growth rate in~$n$ of the $\rho$-th moment of the number of guesses
$\E[G_n(X^n \mid f_n(Y^n))^\rho]$:
\begin{equation}\label{eq:to_characterize}
  \lim_{n \to \infty} \frac{1}{n} \log\Bigl(\min_{f_n} \min_{\set{G}_n}
  \E[G_n(X^n \mid f_n(Y^n))^\rho]\Big)
\end{equation}
(when the limit exists), where the minima in
\eqref{eq:to_characterize} are over all maps $f_n\colon \set{Y}^n
\to \{0, 1\}^{nR}$ and all guessing strategies~$\set{G}_n$.
Theorem~\ref{thm:main} below asserts that the limit exists and
provides a variational
characterization for it.

To state the theorem, we need some additional notation.  Given finite sets
$\set{V}$ and $\set{W}$, let $\pmf(\set{V})$ denote the family of PMFs
on $\set{V}$, and $\pmf(\set{V} \mid \set{W})$ the family of PMFs on
$\set{V}$ indexed by $\set{W}$: for every
$P(\cdot \mid \cdot) \in \pmf(\set{V} \mid \set{W})$ and every
$w \in \set{W}$, we have $P(\cdot \mid w) \in \pmf(\set{V})$. Given
PMFs $P_W \in \pmf(\set{W})$ and
$P_{V \mid W} \in \pmf(\set{V} \mid \set{W})$, we use $P_W P_{V \mid W}$
to denote the joint PMF $P_{W}(w) \> P_{V \mid W}(v \mid w)$ on
$\set{W} \times \set{V}$ (in this context, $P_{V \mid W}(\cdot \mid \cdot)$
is the conditional PMF of $V$ given $W$.)


\begin{theorem}\label{thm:main}
The limit in \eqref{eq:to_characterize} exists and equals
\begin{equation}
  \sup_{Q_Y} \inf_{Q_{U \mid Y}: \Info(Q_{Y;U}) \leq R}
  \sup_{Q_{X \mid YU}} \Big(\rho\RD{d}{D}(Q_{X \mid U}) - \D(Q_{XYU} \| P_{XY}Q_{U
  \mid Y})\Big),\label{eq:target_exponent}
\end{equation}
where the optimization is over $Q_Y \in \pmf(\set{Y})$, $Q_{U \mid Y}
\in \pmf(\set{U} \mid \set{Y})$, $Q_{X \mid YU} \in \pmf(\set{X}
\mid \set{Y} \times \set{U})$, and the choice of the finite set
$\set{U}$; where $\Info(Q_{Y;U})$ is the mutual information between $Y$ and
$U$; $\RD{d}{D}(Q_{X \mid U})$ is the conditional rate-distortion (R-D)
function of $X$ given $U$:
\begin{equation}\label{eq:def_rd_cond}
  \RD{d}{D}(Q_{X \mid U}) \triangleq \min_{Q_{\hat{X} \mid X, U}\colon
  \E[d(X, \hat{X})] \leq D} \Info(Q_{X; \hat{X} \mid U}),
\end{equation}
where $\Info(Q_{X; \hat{X} \mid U})$ is the conditional mutual
information between $X$ and $\hat{X}$ given $U$; and $\D(\cdot \| \cdot)$
denotes relative entropy. All the expressions in \eqref{eq:target_exponent}
are evaluated w.r.t.~to $Q_{XYU} = Q_YQ_{U \mid Y}Q_{X \mid YU}$, and
those in \eqref{eq:def_rd_cond} are w.r.t.~$Q_{\hat{X} \mid X, U}Q_{XU}$,
with $Q_{XU}$ being the $(X, U)$-marginal of $Q_YQ_{U \mid Y}Q_{X \mid YU}$.
\end{theorem}

\begin{remark}
  As shown in Appendix~\ref{app:card}, restricting $U$ to take values
  in a set of cardinality $|\set{Y}| + 1$ does not
  alter~\eqref{eq:target_exponent}. Consequently, the suprema and
  infimum can be replaced by maxima and minimum respectively.
\end{remark}

\begin{remark}
  \label{rem:A}
  In the special case where the help is direct, i.e., when $Y$
  equals~$X$ 
  under $P_{XY}$ so
  \begin{IEEEeqnarray}{rCl}
    \label{eq:direct_help10}
    (x \neq y) & \implies & \bigl(P_{XY}(x,y) = 0\bigr),
  \end{IEEEeqnarray}
  Theorem \ref{thm:main} recovers Theorem 2 of \cite{GraLap18}.
  \end{remark}
  \begin{IEEEproof}[Proof of Remark~\ref{rem:A}]
    This can be seen by first noting that the relative entropy
    in~\eqref{eq:target_exponent} is finite only when
    $Q_{XYU} \ll P_{XY}Q_{U \mid Y}$, whence
    $Q_{XY} \ll P_{XY}$.\footnote{We use $Q \ll P$ to indicate that
      $Q$ is absolutely continuous w.r.t.\ $P$.} This
    and~\eqref{eq:direct_help10} imply that the inner supremum
    in~\eqref{eq:target_exponent} is attained when
  $X$ and $Y$ are equal also under $Q_{XYU}$, and
  \begin{equation}\label{eq:x_equals_y_under_q}
    Q_{X \mid YU}(x \mid y, u) = \I(x = y).
  \end{equation}
  Using \eqref{eq:x_equals_y_under_q} and denoting expectation w.r.t. $Q_{XYU}$
  by $\E_{Q_{XYU}}$, we simplify $\D(Q_{XYU} \| P_{XY}Q_{U \mid Y})$ as
  follows:
  \begin{align}
    \D(Q_{XYU} \| P_{XY}Q_{U \mid Y}) &= \D(Q_YQ_{U \mid Y}Q_{X \mid YU} \| P_{XY}Q_{U \mid Y})\\
    &= \E_{Q_{XYU}}\!\!\Bigg[\log\left(\frac{Q_Y(Y)Q_{U \mid Y}(U \mid Y)Q_{X \mid YU}(X \mid Y, U)}{P_{XY}(X, Y)Q_{U \mid Y}(U \mid Y)}\right)\Bigg]\\
    &= \E_{Q_{XYU}}\!\!\Bigg[\log\left(\frac{Q_Y(Y)Q_{X \mid YU}(X \mid Y, U)}{P_{XY}(X, Y)}\right)\Bigg]\\
    &= \E_{Q_{XYU}}\!\!\Bigg[\log\left(\frac{Q_Y(Y)\I(X = Y)}{P_X(X)\I(Y = X)}\right)\Bigg]\\
    &= \sum_{(x, y) \in \set{X} \times \set{Y}} Q_Y(y)\I(x =
      y)\log\left(\frac{Q_Y(y)\I(x = y)}{P_X(x)\I(y = x)}\right). \label{eq:x_and_y_only}
  \end{align}
  To continue from \eqref{eq:x_and_y_only}, note that, by~\eqref{eq:x_equals_y_under_q},
  \begin{equation}
    Q_Y(y)\I(x = y) = Q_X(x)\I(y = x),
  \end{equation}
  so \eqref{eq:x_and_y_only} implies that
  \begin{align}
    \D(Q_{XYU} \| P_{XY}Q_{U \mid Y}) &= \sum_{(x, y) \in \set{X} \times \set{Y}} Q_X(x)\I(y = x)\log\left(\frac{Q_X(x)\I(y = x)}{P_X(x)\I(y = x)}\right)\\
    &= \sum_{x \in \set{X}} Q_X(x)\log\left(\frac{Q_X(x)}{P_X(x)}\right)\\
    &= \D(Q_X \| P_X).\label{eq:simplified_rel_entropy}
  \end{align}
  Having dispensed with the inner supremum
  in~\eqref{eq:target_exponent}, we note that, because
  $X$ and $Y$ are equal under $Q_{XYU}$, we can replace
  the outer supremum in \eqref{eq:target_exponent} with one over $Q_X$,
  and the infimum with one over $Q_{U \mid X}$. From this and
  \eqref{eq:simplified_rel_entropy} we conclude that
  \eqref{eq:target_exponent} reduces to
  \begin{equation}
    \sup_{Q_X} \inf_{Q_{U \mid X}: \Info(Q_{X;U}) \leq R} \Big(\rho\RD{d}{D}(Q_{X \mid U}) - \D(Q_X \| P_X)\Big),\label{eq:target_exponent_x_equals_y}
  \end{equation}
  which recovers Theorem 2 of \cite{GraLap18}.
\end{IEEEproof}

\begin{remark}\label{rem:B}
  When the help is useless because $R$ is zero or because $X$ and $Y$ are
  independent (under $P_{XY}$), Theorem 1 reduces to Corollary 1 of
  \cite{ArikanMerhav98}.
  \end{remark}
  \begin{IEEEproof}[Proof of Remark~\ref{rem:B}]
  To show this, we begin by considering the choice of $U$ as
  deterministic and thus establish that~\eqref{eq:target_exponent} is upper
  bounded by
  \begin{equation}
    \label{eq:netflix10}
    \sup_{Q_X} \Big( \rho\RD{d}{D}(Q_X) - \D(Q_{X} \| P_{X}) \Big),
  \end{equation}
  which is the expression in Corollary 1 of \cite{ArikanMerhav98}. It
   remains to show that, when $R = 0$ or when $X$ and $Y$ are
   independent, this is also a lower bound.

   We begin with $R=0$. In this case, the constraint in the infimum
   in~\eqref{eq:target_exponent} implies that $Y$ and $U$ are
   independent under $Q_{XYU}$, so
  \begin{equation}\label{eq:Ay_and_u_independent}
    Q_{XYU} = Q_Y \, Q_U \, Q_{X \mid YU}.
  \end{equation}
  A lower bound results when we restrict the inner supremum to
  $Q_{X|YU}(x|y,u)$ that is determined by $x$ and $y$, so that $Q_{XYU}$
  has the form $Q_{U}Q_{XY}$. With this form, the objective function
  in~\eqref{eq:target_exponent} reduces to 
  \begin{equation}
    \rho\RD{d}{D}(Q_{X}) - \D(Q_{XY} \| P_{XY})
  \end{equation}
  which depends on $Q_{XYU}$ only via its marginal $Q_{XY}$. This
  allows us to dispense with the infimum to obtain
  \begin{equation}
    \label{eq:netflix50}
    \sup_{Q_{XY}} \Big(\rho\RD{d}{D}(Q_{X}) - \D(Q_{XY} \| P_{XY})\Big),
  \end{equation}
  which is attained when $Q_{Y|X}$ equals $P_{Y|X}$, whence it is
  equal to~\eqref{eq:netflix10}. 

  Having established that~\eqref{eq:netflix10} is a lower bound
  on~\eqref{eq:target_exponent} when $R=0$, we now show that it is
  also a lower bound on~\eqref{eq:target_exponent} when $X$ and $Y$
  are independent. In this case we obtain the lower bound by
  restricting the inner supremum to be over $Q_{X|YU}(x|y,u)$ that are
  determined by $x$ alone, so that $Q_{XYU}$ has the form
  $Q_{X}Q_{UY}$. With this form (and with $X$ and $Y$ being
  independent under $P_{XY}$), the objective function
  in~\eqref{eq:target_exponent} reduces to
  \begin{equation}
    \rho\RD{d}{D}(Q_{X}) - \D(Q_{X}Q_{YU} \| P_{X}P_{Y} Q_{U|Y})
  \end{equation}
  which simplifies to
  \begin{equation}
    \rho\RD{d}{D}(Q_{X}) - \D(Q_{X}Q_{Y} \| P_{X}P_{Y}).
  \end{equation}
  Again $U$ disappears, and we are back at~\eqref{eq:netflix50}, which
  evaluates to the desired lower bound.
  \end{IEEEproof}

\section{Achievability}
In this section, we prove the direct part of Theorem \ref{thm:main},
namely, that when $\{(X_i, Y_i)\}_{i = 1}^n$ are IID according to
$P_{XY}$, then for every $\epsilon > 0$ there exists a sequence
of rate-$R$ helpers $\{f_n\}$ and guessing strategies $\{\set{G}_n\}$
satisfying
\begin{align}
  &\limsup_{n \to \infty} \frac{1}{n}\log(\E[G_n(X^n
  \mid f_n(Y^n))^\rho])\nonumber\\
  &\quad\quad \leq \sup_{Q_Y} \inf_{Q_{U \mid Y}: \Info(Q_{Y;U}) \leq R}
  \sup_{Q_{X \mid YU}} \Big(\rho\RD{d}{D}(Q_{X \mid U}) - \D(Q_{XYU} \| P_{XY}Q_{U \mid Y})\Big)
  + \epsilon.\label{eq:target_ub}
\end{align}

\begin{proof}
Since we are only interested in the behavior of $\E[G_n(X^n
\mid f_n(Y^n))^\rho]$ as $n$ tends to infinity, we shall only consider
large values of $n$.

We begin by constructing the helper $f_n$. To do so, we
shall use the Type-Covering lemma \cite[Lemma 1]{berger71}, \cite[Lemma 9.1]{csiszar_korner_11}, \cite[Lemma 2.34]{Moser21} that we
restate here for the reader's convenience. Given finite sets $\set{V}$
and $\set{W}$, let $\pmf_n(\set{V})$ denote the family of
``types of denominator $n$'' on $\set{V}$, i.e., the PMFs $P(\cdot)
\in \pmf(\set{V})$ for which $nP(v)$ is an integer for all $v \in \set{V}$.
By a ``conditional type on $\set{V}$ given $\set{W}$'' we refer to a
conditional PMF $P(\cdot \mid \cdot) \in \pmf(\set{V} \mid \set{W})$ for
which $P(\cdot \mid w)$ is a type (of some denominator $n(w)$) for every $w
\in \set{W}$. Given a sequence $v^n \in \set{V}^n$, the ``empirical distribution
of $v^n$'' is the (unique) type $P \in \pmf_n(\set{V})$ for which $P(v') =
\frac{1}{n}|\{i\colon v_i = v'\}|$ for every $v' \in \set{V}$. And given $P
\in \pmf_n(\set{V})$, we use $\tcl{n}(P)$ to denote the ``type class'' of
$P$, i.e., the set of all sequences $v^n \in \set{V}^n$ whose empirical
distribution is $P$.

\begin{lemma}[Type-Covering lemma]\label{lmm:tcl}
Let $\set{V}$ and $\set{W}$ be finite sets. For every $\epsilon > 0$
there exists some $n_0(\epsilon)$ such that for all $n$ exceeding
$n_0(\epsilon)$ the following holds: For every $Q_V \in \pmf_n(\set{V})$
and every conditional type $Q_{W \mid V}$ for which $Q_VQ_{W \mid V}
\in \pmf_n(\set{V} \times \set{W})$,
there exists a codebook $\set{C} \subseteq \set{W}^n$ satisfying
\begin{equation}
  |\set{C}|  \leq  2^{n(\Info(Q_{V; W}) + \epsilon)}\label{eq:tcl_prop_1}
\end{equation}
and
\begin{equation}
  \forall v^n\! \in\! \tcl{n}(Q_V) \, \exists w^n\!\!
  \in\! \set{C}\colon\! (v^n, w^n) \in \tcl{n}(Q_VQ_{W
  \mid V}).\label{eq:tcl_prop_2}
\end{equation}
\end{lemma}

Lemma \ref{lmm:tcl} is applied as follows: For every
$Q_Y \in \pmf_n(\set{Y})$, we first define
\begin{align}\label{eq:optimal_cond}
  Q_{U \mid Y}^*(Q_Y) &\triangleq \argmin_{Q_{U \mid Y}:
  \Info(Q_{U; Y}) \leq R - \epsilon'} \max_{Q_{X
  \mid YU}} \RD{d}{D}(Q_{X \mid U}),
\end{align}
(provided the minimum exists) where the optimization is over choice of
the finite set $\set{U}$, and types $Q_{U \mid Y}$ and $Q_{X \mid YU}$
for which $Q_YQ_{U \mid Y}Q_{X \mid YU} \in \pmf_n(\set{Y} \times \set{U}
\times \set{X})$; where $\Info(Q_{U; Y})$ and $\RD{d}{D}(Q_{X \mid U})$
are computed w.r.t.~$Q_YQ_{U \mid Y}Q_{X \mid YU}$; and where $\epsilon'$
is a small positive constant (to be specified later). If the minimum in
\eqref{eq:optimal_cond} does not exist, we let
\begin{equation}
  R^*(Q_Y) \triangleq \inf_{Q_{U \mid Y}:
  \Info(Q_{U; Y}) \leq R - \epsilon'} \max_{Q_{X
  \mid YU}} \RD{d}{D}(Q_{X \mid U}),
\end{equation}
where the optimization is under the same conditions as in
\eqref{eq:optimal_cond}, and instead define $Q_{U \mid Y}^*(Q_Y)$ as a
conditional type satisfying
\begin{equation}\label{eq:r_star}
  \max_{Q_{X \mid YU}} \RD{d}{D}(Q_{X \mid U}) \leq R^*(Q_Y) + \epsilon''
\end{equation}
where the maximum is over all conditional types $Q_{X \mid YU}$ for
which
$Q_YQ_{U \mid Y}^*Q_{X \mid YU} \in \pmf_n(\set{Y} \times \set{U}
\times \set{X})$; where $\RD{d}{D}(Q_{X \mid U})$ is computed
w.r.t.~$Q_YQ_{U \mid Y}^*Q_{X \mid YU}$; and where $\epsilon''$ is a
small positive constant (also to be specified later).

To construct the helper $f_n$, we invoke Lemma \ref{lmm:tcl}
(assuming that $n$ is sufficiently large) setting $Q_V \leftarrow Q_Y$,
$Q_{W \mid V} \leftarrow Q_{U \mid Y}^*(Q_Y)$, and $\epsilon \leftarrow
\epsilon'$ to obtain a codebook $\set{C}(Q_Y) \subseteq \set{U}^n$ used
by $f_n$ to produce the index of some $U^n \in \set{C}(Q_Y)$ such that
$(U^n, Y^n) \in \tcl{n}(Q_YQ_{U \mid Y}^*(Q_Y))$.

We next construct a guessing strategy $\set{G}_n$. Let $U^n
\in \set{C}(Q_Y)$ be the codeword provided by the helper and that
hence satisfies $(Y^n, U^n) \in \tcl{n}(Q_YQ_{U \mid Y}^*(Q_Y))$.
Let $Q_{XYU}$ denote the empirical joint distribution of $(X^n, Y^n, U^n)$.
We first argue that the guesser can be assumed cognizant of $Q_{XYU}$.
To that end, we need the following lemma:

\begin{lemma}[{Interlaced-Guessing lemma \cite[Lemma 5]{BraHofLap19}}]\label{lmm:int_guessing}
Let $V$, $W$, and $Z$ be finite-valued chance variables and
let $\rho$ be nonnegative. Given any guessing strategy $\set{G}$ for
guessing $V$ based on $W$ and $Z$, there exists a guessing strategy
$\Tilde{\set{G}}$ based on $W$ only such that
\begin{equation}
  \E[\tilde{G}(V \mid W)^\rho] \leq \E[G(V \mid W, Z)^\rho]
  |\set{Z}|^\rho.
\end{equation}
\end{lemma}

Invoking Lemma \ref{lmm:int_guessing} with $V \leftarrow X^n$, $W
\leftarrow U^n$, and $Z \leftarrow Q_{XYU}$, we see that
\begin{equation}
  \min_{\set{G}_n} \E[G(X^n \mid U^n)^\rho] \leq \min_{\set{G}_n} \E[G(X^n \mid U^n,
  Q_{XYU})^\rho] \big|\!\typ{n}(\set{X} \times \set{Y}
  \times \set{U})\big|^\rho\!\!\!,\label{eq:int_guessing_ub}
\end{equation}
where the guessing strategy on the RHS of \eqref{eq:int_guessing_ub}
depends on both the helper's description $f_n(Y^n)$ of $Y^n$ and the
empirical joint distribution $Q_{XYU}$ of $(X^n, Y^n, U^n)$. Since
$\big|\!\typ{n}(\set{X} \times \set{Y} \times \set{U})\big|$ grows subexponentially
with $n$,
\begin{equation}\label{eq:poly_ub}
  \lim_{n \to \infty} \frac{1}{n}\log\big|\!\typ{n}(\set{X}
  \times \set{Y} \times \set{U})\big|^\rho = 0.
\end{equation}
Thus, by
\eqref{eq:int_guessing_ub} and \eqref{eq:poly_ub},
\begin{equation}
  \limsup_{n \to \infty} \frac{1}{n}\min_{\set{G}_n} \E[G(X^n
  \mid U^n)^\rho] \leq \limsup_{n \to \infty} \frac{1}{n}\min_{\set{G}_n} \E[G(X^n
  \mid U^n, Q_{XYU})^\rho].\label{eq:type_revealed}
\end{equation}
Since the RHS of \eqref{eq:type_revealed} cannot exceed its LHS,
\eqref{eq:type_revealed} must hold with equality, and we shall hence
for the remainder of the proof assume that $Q_{XYU}$ is known to the
guesser.

Our guessing strategy $\set{G}_n$ will thus depend
on both the helper's description $U^n$ of $Y^n$ and the empirical joint
distribution $Q_{XYU}$ of $(X^n, Y^n, U^n)$. To construct $\set{G}_n$, we
will use of the following corollary \cite[Lemma 2]{GraLap18} which follows
from the conditional version of Lemma \ref{lmm:tcl}:

\begin{corollary}\label{corr:cond_rd_tcl}
Let $\set{V}$, $\set{W}$ and $\set{Z}$ be finite
sets, let $\dist(\cdot, \cdot)$ be a distortion function on $\set{V}
\times \set{W}$, let $\bar{\dist}(\cdot, \cdot)$ be its extension to
sequences, and let $D$ be positive. For
every $\epsilon > 0$ there exists some $n_0(\epsilon)$ such that
for all $n$ exceeding $n_0(\epsilon)$ the following holds: For every
$Q_{VZ} \in \pmf_n(\set{V} \times \set{Z})$ and every $z^n
\in \tcl{n}(Q_Z)$ there exists a codebook $\set{C} \subseteq
\set{W}^n$ that satisfies
\begin{equation}
  |\set{C}| \leq 2^{n(\RD{d}{D}(Q_{V \mid Z})
  + \epsilon)}\label{eq:rd_cond_tcl_prop_1}
\end{equation}
and
\begin{equation}
  \forall v^n \in \tcl{n}(Q_{V \mid Z} | z^n) \, \exists w^n
  \in \set{C}\colon \bar{\dist}(v^n, w^n)
  \leq D.\label{eq:rd_cond_tcl_prop_2}
\end{equation}
\end{corollary}

We invoke Corollary \ref{corr:cond_rd_tcl} with $Q_{VZ} \leftarrow Q_{XU}$,
$z^n \leftarrow U^n$, $\set{W} \leftarrow \hat{\set{X}}$, and $\epsilon
\leftarrow \epsilon''$, where $\epsilon''$ is some small nonnegative
constant (to be specified later) to obtain the codebook $\set{C}(Q_{XYU})
\subseteq \hat{\set{X}}^n$. The guessing strategy $\set{G}_n$ is then
chosen such that $\set{G}_n|_{\{1, \ldots, |{\set{C}(Q_{XYU})}|\}}$
is a bijection from $\{1, \ldots, |\set{C}(Q_{XYU})|\}$ to
$\set{C}(Q_{XYU})$, i.e., such the first $|\set{C}(Q_{XYU})|$ guesses are
those in $\set{C}(Q_{XYU})$ in some arbitrary order. Note that
\eqref{eq:rd_cond_tcl_prop_2} guarantees that some $\hat{X}^n$ in~$\set{G}_n|_{\{1, \ldots, |{\set{C}(Q_{XYU})}|\}}$ satisfies
\eqref{eq:thres_to_sat}, and thus the guesser succeeds after at
most $|\set{C}(Q_{XYU})|$ guesses.

We now show that \eqref{eq:target_ub} holds for our proposed helper $f_n$
and guessing strategy $\set{G}_n$:
\begin{align}
  &\E[G_n(X^n \mid f_n(Y^n))^\rho]\nonumber\\
  &\stackrel{(a)}{=} \E[G_n(X^n \mid U^n)^\rho]\\
  &\stackrel{(b)}{=} \sum_{Q_Y} \sum_{Q_{X \mid YU}} \Big(\!
  \Prb[Y^n \in \tcl{n}(Q_Y)] \Prb[X^n \in \tcl{n}(Q_{X \mid YU}) \mid Y^n
  \in \tcl{n}(Q_Y)]\nonumber\\
  &\quad\quad\E[G_n(X^n \mid U^n)^\rho \mid (X^n, Y^n, U^n)
  \in \tcl{n}(Q_YQ_{U \mid Y}^*(Q_Y)Q_{X \mid YU})]\Big)\\
  &\stackrel{(c)}{\leq} \sum_{Q_Y} \sum_{Q_{X \mid YU}}
  \Prb[Y^n \in \tcl{n}(Q_Y)] \Prb[X^n \in \tcl{n}(Q_{X \mid YU}) \mid Y^n
  \in \tcl{n}(Q_Y)] 2^{n\rho(\RD{d}{D}(Q_{X \mid U})) + \epsilon'')}\\
  &\stackrel{(d)}{\leq} \sum_{Q_Y} \sum_{Q_{X \mid YU}} 2^{-n\D(Q_Y \| P_Y)}2^{-n\D(Q_{X \mid YU} \| P_{X
  \mid Y})} 2^{n\rho(\RD{d}{D}(Q_{X \mid U}) + \epsilon'')}\\
  &\stackrel{(e)}{\leq} \max_{Q_Y} \max_{Q_{X \mid YU}}
  2^{-n\D(Q_Y \| P_Y)}2^{-n\D(Q_{X \mid YU} \| P_{X
  \mid Y})} 2^{n\rho(\RD{d}{D}(Q_{X \mid U}) + \epsilon'')} \big|\!\typ{n}(\set{X} \times \set{Y}
  \times \set{U})\big|^\rho\\
  &\stackrel{(f)}{=} \max_{Q_Y} \min_{Q_{U \mid Y}: \Info(Q_{U; Y})
  \leq R - \epsilon'} \max_{Q_{X \mid YU}}
  2^{-n\D(Q_Y \| P_Y)} 2^{-n\D(Q_{X \mid YU} \| P_{X \mid Y})}2^{n\rho(\RD{d}{D}(Q_{X
  \mid U}) + \epsilon'')} \big|\!\typ{n}(\set{X} \times \set{Y}
  \times \set{U})\big|^\rho\\
  &\stackrel{(g)}{\leq} \sup_{Q_Y} \inf_{Q_{U \mid Y}:
  \Info(Q_{U; Y}) \leq R - \epsilon'} \sup_{Q_{X \mid YU}}
  2^{-n\D(Q_Y \| P_Y)}2^{-n\D(Q_{X \mid YU} \| P_{X \mid Y})}2^{n\rho(\RD{d}{D}(Q_{X
  \mid U}) + \epsilon'')}2^{n\delta_n} \big|\!\typ{n}(\set{X} \times \set{Y}
  \times \set{U})\big|^\rho\\
  &\stackrel{(h)}{\leq} \sup_{Q_Y} \inf_{Q_{U \mid Y}:
  \Info(Q_{U; Y}) \leq R} \sup_{Q_{X \mid YU}}
  2^{-n\D(Q_Y \| P_Y)} 2^{-n\D(Q_{X \mid YU} \| P_{X \mid Y})}2^{n\rho\RD{d}{D}(Q_{X
  \mid U})}2^{n\delta_n}2^{n\epsilon}
  \big|\!\typ{n}(\set{X} \times \set{Y}
  \times \set{U})\big|^\rho\label{eq:final_step},
\end{align}
where (a) holds because we have assumed that the empirical distribution
$Q_Y$ of $Y^n$ is known to the guesser who can thus recover $U^n$ from
$f_n(Y^n)$ and $\set{C}(Q_Y)$; in (b) we have used the law of total
expectation, averaging over the types $Q_Y \in \pmf_n(\set{Y})$ and
conditional types $Q_{X \mid YU}$ for which $Q_YQ_{U \mid Y}Q_{X
\mid YU} \in \pmf_n(\set{Y} \times \set{U} \times \set{X})$ (recall
that $Q_{U \mid Y} = Q_{U \mid Y}^*(Q_Y)$ is fixed by $f_n$); (c) is due
to \eqref{eq:rd_cond_tcl_prop_1}; (d) follows from \cite[Theorem 11.1.4]{CovTho06};
in (e) we have upper-bounded the sum by the largest term times the
number of terms (the number of terms is the number of types~$Q_Y$ and
$Q_{X \mid YU}$ that we have in turn upper-bounded by the number of
types $Q_{XYU}$);
(f) is due to \eqref{eq:optimal_cond}; in (g) we have lifted the
constraint on $Q_Y$, $Q_{U \mid Y}^*(Q_Y)$, and $Q_{X \mid YU}$ to be
types at a cost of at most $2^{n\delta_n}$, where $\delta_n
\downarrow 0$ as $n \to \infty$, and where the step is justified
because any PMF can be approximated
arbitrarily well by a type of sufficiently large denominator; and in (h)
we have used the fact that all exponents are continuous functions of their
respective arguments, and that $\epsilon'$ and $\epsilon''$
were chosen sufficiently small.

Dividing the $\log$ of \eqref{eq:final_step} by $n$, taking
the $\limsup$ as $n$ tends to infinity, and applying \eqref{eq:poly_ub}
yields \eqref{eq:target_ub}.
\end{proof}

\section{Converse}

In this section we prove the converse part of Theorem \ref{thm:main},
namely, that when $\{(X_i, Y_i)\}_{i = 1}^n$ are IID according to
$P_{XY}$, then for any sequence of rate-$R$ helpers $\{f_n\}$ and
guessing strategies $\{\set{G}_n\}$,
\begin{align}
  &\liminf_{n \to \infty} \frac{1}{n}\log(\E[G_n(X^n
  \mid f_n(Y^n))^\rho])\nonumber\\
  &\quad\quad \geq \sup_{Q_Y} \inf_{Q_{U \mid Y}: \Info(Q_{Y;U}) \leq R}
  \sup_{Q_{X \mid YU}} \Big(\rho\RD{d}{D}(Q_{X \mid U}) - \D(Q_{XYU} \| P_{XY}Q_{U
  \mid Y})\Big).\label{eq:target_lb}
\end{align}

\begin{proof}
Fix a sequence of helpers $\{f_n$\} and guessing strategies
$\{\set{G}_n\}$. We begin by observing that for any probability
law $Q$ of $(X^n, Y^n)$-marginal $Q_{X^nY^n}$,
\begin{equation}
  \E_{P_{X^nY^n}}[G_n(X^n \mid f_n(Y^n))^\rho] \geq 2^{\rho\E_Q[\log(G_n(X^n \mid f_n(Y^n)))]
  - \D(Q_{X^nY^n} \| P_{X^nY^n})}\label{eq:exp_lb},
\end{equation}
where $\E_P$ denotes expectation w.r.t.\ the PMF $P$. Indeed,
\begin{align}
  \E_{P_{X^nY^n}}[G_n(X^n \mid f_n(Y^n))^\rho] &= \sum_{(x^n, y^n) \in \set{X}^n \times \set{Y}^n} P_{X^n, Y^n}(x^n, y^n)\, G_n(x^n
  \mid f_n(y^n))^\rho\\
  &= \sum_{(x^n, y^n) \in \set{X}^n \times \set{Y}^n} Q_{X^n, Y^n}(x^n, y^n)\, G_n(x^n
  \mid f_n(y^n))^\rho\frac{P_{X^n, Y^n}(x^n, y^n)}
  {Q_{X^n, Y^n}(x^n, y^n)}\\
  &= \sum_{(x^n, y^n) \in \set{X}^n \times \set{Y}^n} Q_{X^n, Y^n}(x^n, y^n)\,
  2^{\log\left(G_n(x^n \mid f_n(y^n))^\rho
  \frac{P_{X^n, Y^n}(x^n, y^n)}{Q_{X^n, Y^n}(x^n, y^n)}\right)}\\
  &\stackrel{(a)}{\geq} 2^{\sum_{x^n, y^n} Q_{X^n, Y^n}\log\left(
  G_n(x^n \mid f_n(y^n))^\rho\frac{P_{X^n, Y^n}(x^n, y^n)}
  {Q_{X^n, Y^n}(x^n, y^n)}\right)}\\
  &= 2^{\rho\E_Q[\log(G_n(X^n \mid f_n(Y^n)))]
  - \D(Q_{X^nY^n} \| P_{X^nY^n})},
\end{align}
where (a) follows from Jensen's inequality.

To describe the law $Q$ to which we shall apply \eqref{eq:exp_lb},
let $[1:n]$ denote the set $\{1,\ldots, n\}$ and define the auxiliary
variables
\begin{align}
  M   &\triangleq f_n(Y^n)\\
  U_i &\triangleq (X^{i - 1}, Y^{i - 1}, M), \quad i \in [1:n]
  \label{eq:u_i}
\end{align}
taking values in the sets
\begin{equation}
  \set{M} \triangleq \{0, 1\}^{nR}
\end{equation}
and
\begin{equation}
  \label{eq:defSetUi}
  \set{U}_i \triangleq \set{X}^{i - 1} \times \set{Y}^{i - 1}
                       \times \set{M},
  \quad i \in [1:n].
\end{equation}
Given any $Q_Y \in \pmf(\set{Y})$ and any $n$ Markov kernels $\{Q_{X_i
\mid Y_iU_i}\}_{i = 1}^n$, define the law
\begin{equation}
  Q_{X^nY^nMU^n\hat{X}^n} \textnormal{ on } {\set{Y}^n \times \set{X}^n \times \set{M} \times \prod_{i = 1}^n
\set{U}_i \times \hat{\set{X}}^n}
\end{equation}
as
\begin{subequations}
  \begin{align}
    &Q_{X^nY^nMU^n\hat{X}^n} \triangleq Q_Y^{\times n}P_{M \mid Y^n}\prod_{i
    = 1}^n \big(Q_{U_i \mid X^{i - 1}Y^{i - 1}M}Q_{X_i \mid
    Y_iU_i}\big)P_{\hat{X}^n \mid MX^n},\label{eq:q_full}
  \end{align}
  where $P_{M \mid Y^n}$ is specified by the helper as
  \begin{equation}
    P_{M \mid Y^n}(m \mid y^{n}) = \I(m = f_n(y^n)),
  \end{equation}
  $Q_{U_i \mid X^{i - 1}Y^{i - 1}M}$ is specified
  through the definition of $U_{i}$ in~\eqref{eq:u_i} as
  \begin{equation}
    Q_{U_i \mid X^{i - 1}Y^{i - 1}M}(u_{i} \mid x^{i - 1}y^{i - 1}m) = \I(u_i = (x^{i - 1}, y^{i - 1}, m)),
  \end{equation}
  and 
  $P_{\hat{X}^n \mid MX^n}$ is determined by the guessing strategy as
  \begin{equation}
   P_{\hat{X}^n \mid MX^n}(\hat{x}^{n} \mid m, x^{n}) =  \I(\hat{x}^n = \set{G}_n(G_n(x^n \mid m))).
 \end{equation}
 Thus,
  \begin{align}
    &Q_{X^nY^nMU^n\hat{X}^n}(y^n, x^n, m, u^n, \hat{x}^n)\nonumber\\
    &\quad\quad= Q_Y^{\times n}(y^n)\I(m = f_n(y^n))
    \prod_{i = 1}^n \big(\!\I(u_i = (x^{i - 1}, y^{i - 1}, m))Q_{X_i \mid
    Y_iU_i}(x_i \mid y_i, u_i)\big)\I(\hat{x}^n = \set{G}_n(G_n(x^n \mid m))),
  \end{align}
\end{subequations}
where $G_n(\cdot)$ is defined in \eqref{eq:def_guessing_fn}.

Note that~\eqref{eq:q_full} implies that
\begin{equation}
  \label{eq:amos_clar700Markov}
  X^{i-1} \to  (M,Y^{i-1}) \to Y_{i} \quad \text{under $Q$}
\end{equation}
because the $(X^{i-1},Y^{n},M,U^{i-1})$-marginal of $Q$ can be written
as
\begin{equation}
  Q_{X^{i-1}Y^{n}MU^{i-1}} = Q_Y^{\times n} P_{M \mid Y^n} \prod_{j
    = 1}^{i-1} \big(Q_{U_j \mid X^{j - 1}Y^{j - 1}M}Q_{X_j \mid
    Y_jU_j}\big),
\end{equation}
which impies that
\begin{equation}
  \label{eq:amos_clar710Markov}
  (X^{i-1},U^{i-1}) \to  (M,Y^{i-1}) \to Y_{i}^{n} \quad \text{under $Q$}
\end{equation}
because the product is a function of $(m,y^{i-1})$ and
$(x^{i-1},u^{i-1})$, and the pre-product
$Q_Y^{\times n}(y^n) P_{M \mid Y^n}$ is a function of $(m,y^{i-1})$
and $y_{i}^{n}$.

Next define for every $i \in [1:n]$
\begin{equation}\label{eq:dist_i}
  D_i \triangleq \E[d(X_i, \hat{X}_i)],
\end{equation}
where the expectation is w.r.t.\ to the PMF
  $Q_{X^nY^nMU^n\hat{X}^n}$. Under the latter, $\hat{x}^n =
  \set{G}_n(G_n(x^n \mid m))$ and therefore $\bar{\dist}(x^n,
  \hat{x}^n) \leq D$ for \emph{every} $x^{n} \in \set{X}^{n}$ and,
  also in expectation (over $Q_{X^nY^nMU^n\hat{X}^n}$)
  \begin{equation}
    \label{eq:amos_clar200}
    \frac{1}{n} \sum_{i=1}^{n} D_{i} \leq D.
  \end{equation}

Further define 
\begin{equation}\label{eq:min_q_i}
  Q_{\hat{X}_i' \mid MX^i}^* \triangleq \argmin_{\substack{Q_{\hat{X}'_i
  \mid MX^i}:\\ \E[d(X_i, \hat{X}'_i)] \leq D_i}} \Info(X_i; \hat{X}'_i
  \mid M, X^{i - 1}),
\end{equation}
where the minimum is over all conditional PMFs
$Q_{\hat{X}'_i \mid MX^i} \in \pmf(\hat{\set{X}} \mid \set{M} \times
\set{X}^i)$, and where $\Info(X_i; \hat{X}'_i \mid M, X^{i - 1})$ and
$\E[d(X_i, \hat{X}'_i)]$ are evaluated
w.r.t.~$Q_{\hat{X}'_i \mid MX^i} Q_{MX^i}$, with $Q_{MX^i}$ being the
$(M, X^i)$-marginal of $Q_{X^nY^nMU^n\hat{X}^n}$. Using
$\{Q_{\hat{X}_i' \mid MX^i}^*\}_{i = 1}^n$, we extend
$Q_{X^nY^nMU^n\hat{X}^n}$ to a law $Q$ on
$\set{Y}^n \times \set{X}^n \times \set{M} \times \prod_{i = 1}^n
\set{U}_i \times \hat{\set{X}}^n \times \hat{\set{X}}^n$ as follows:
\begin{equation}\label{eq:total_q}
  Q \triangleq Q_{X^nY^nMU^n\hat{X}^n}\prod_{i = 1}^n Q_{\hat{X}'_i \mid
  MX^i}^*.
\end{equation}
Note that the factorization in~\eqref{eq:total_q} implies that
  \begin{equation}
    \label{eq:amos_clar_10}
    \hat{X}'_i \to (M, X^i) \to Y^{i - 1}
  \end{equation}
  because it implies that---conditional on $(M,X^i)$---$\hat{X}'_i$ is
  independent of the tuple $(X^n,Y^n,M,U^n,\hat{X}^n)$ and hence also of $Y^{i-1}$
  (which is a function of this tuple). 
%
%
For the remainder of this section we shall assume that, unless stated
otherwise, all expectations and information-theoretic quantities are
evaluated w.r.t.~$Q$. To study \eqref{eq:exp_lb} for this~$Q$, we
begin by lower-bounding $\E[\log(G_n(X^n \mid M))]$ using the
conditional R-D function. To this end, we note that, conditional on
$M=m$, there is a one-to-one correspondence between $G_n(X^n \mid M)$ and
$\hat{X}^{n}$ so, by the Reverse Wyner inequality of Corollary~\ref{cor:entropy_lb}
in Appendix \ref{app:a},
\begin{equation}
  \E[\log(G_n(X^n \mid M)) \mid M=m] \geq \Entr(\hat{X}^n \mid M = m) - n\delta_n\\
\end{equation}
with $\delta_n$ tending to zero as $n$ tends to infinity. Averaging
over $M$,
\begin{align}
  \E[\log(G_n(X^n \mid M))] &\geq \Entr(\hat{X}^n \mid M) - n\delta_n\\
  &\geq \Info(\hat{X}^n; X^n \mid M) - n\delta_n\\
  &= \sum_{i = 1}^n \Big(\!\Entr(X_i \mid M, X^{i - 1})
  - \Entr(X_i \mid M, \hat{X}^n, X^{i - 1})\Big) - n\delta_n\\
  &\geq \sum_{i = 1}^n \Big(\!\Entr(X_i \mid M, X^{i - 1})
  - \Entr(X_i \mid M, \hat{X}_i, X^{i - 1})\Big) - n\delta_n\\
  &= \sum_{i = 1}^n \Info(X_i; \hat{X}_i \mid M, X^{i - 1})
  - n\delta_n\\
  &\stackrel{(a)}{\geq} \sum_{i = 1}^n \Info(X_i; \hat{X}'_i
  \mid M, X^{i - 1}) - n\delta_n\\
  &= \sum_{i = 1}^n \Big(\!\Entr(\hat{X}'_i \mid M, X^{i - 1})
  - \Entr(\hat{X}'_i \mid M, X^i)\Big) - n\delta_n\\
  &\geq \sum_{i = 1}^n \Big(\!\Entr(\hat{X}'_i \mid M, X^{i - 1},
  Y^{i -1}) - \Entr(\hat{X}'_i \mid M, X^i)\Big) - n\delta_n\\
  &\stackrel{(b)}{=} \sum_{i = 1}^n \Big(\!\Entr(\hat{X}'_i
  \mid M, X^{i - 1}, Y^{i -1}) - \Entr(\hat{X}'_i
  \mid M, X^i, Y^{i - 1})\Big) - n\delta_n\\
  &\stackrel{(c)}{=} \sum_{i = 1}^n \Big(\!\Entr(\hat{X}'_i \mid U_i)
  - \Entr(\hat{X}'_i \mid U_i, X_i)\Big) - n\delta_n\\
  &= \sum_{i = 1}^n \Info(X_i; \hat{X}'_i \mid U_i) -
  n\delta_n,\label{eq:anchor}
\end{align}
where in (a) we have replaced $\hat{X}_i$ by $\hat{X}'_i$, and
the inequality hence follows from \eqref{eq:min_q_i}; (b) follows
  from~\eqref{eq:amos_clar_10}; 
and in (c) we have identified the auxiliary variable $U_i$ defined in
\eqref{eq:u_i}. To continue from \eqref{eq:anchor}, let $T$ be
equiprobable over $[1:n]$, independent of
$(Y^n, M, X^n, U^n, (\hat{X}')^n)$, and define the chance variable
\begin{equation}
  \label{eq:amos_clar_500}
  (Y, X, U, \hat{X}') \triangleq (Y_T, X_T, U_T, \hat{X}'_T)
\end{equation}
taking values in the set
$\set{Y} \times \set{X} \times (\cup_{i = 1}^n \set{U}_i) \times
\hat{\set{X}}$. Note that, since the sets $\{\set{U}_{i}\}$
of~\eqref{eq:defSetUi} are disjoint, $T$ is a deterministic function
of $U$, and we can define $\iota(\cdot)$ as mapping each
$u \in \cup_{i = 1}^n \set{U}_i$ to the unique $i \in [1:n]$ for which
$u \in \set{U}_i$. With this definition, the PMF of $(Y, X, U,
\hat{X}')$ can be expressed as
\begin{equation}\label{eq:q_mix}
  \tilde{Q}_{YXU\hat{X}'}(y, x, u, \hat{x}') \triangleq \frac{1}{n}Q_{Y_{\iota(u)}X_{\iota(u)}
  U_{\iota(u)}\hat{X}'_{\iota(u)}}(y, x, u, \hat{x}'),
\end{equation}
where $Q_{Y_iX_iU_i\hat{X}'_i}$ is the
$(Y_i, X_i, U_i, \hat{X}'_i)$-marginal of $Q$. We next observe that,
under $\tilde{Q}$, $\E[\dist(X, \hat{X}')]$ is upper-bounded by
$D$. Indeed,
\begin{align}
  \E_{\tilde{Q}}[\dist(X, \hat{X}')] &= \frac{1}{n}\sum_{i = 1}^n
  \E_Q[\dist(X_i, \hat{X}'_i)]\\
  &\leq \frac{1}{n}\sum_{i = 1}^n D_i\\
  &\leq D,\label{eq:dist_upper_bound_q_tilde}
\end{align}
where the first inequality follows from the constraint in the optimization
on the RHS of \eqref{eq:min_q_i} and the second from~\eqref{eq:amos_clar200}.
Also note that, since $T$ is a deterministic function of $U$, the RHS
of~\eqref{eq:anchor} can be expressed in terms of~$(Y, X, U, \hat{X}')$ as
\begin{equation}
  n\Info(X; \hat{X}' \mid U) - n\delta_n,
\end{equation}
so,
\begin{equation}\label{eq:tmp_inf_lb}
  \E_{Q}[\log(G_n(X^n \mid M))] \geq n\Info(X; \hat{X}' \mid U) - n\delta_n,
\end{equation}
where
the conditional mutual information on the RHS is
w.r.t.~$\tilde{Q}$. Using~\eqref{eq:dist_upper_bound_q_tilde}, we can
lower-bound the RHS of
\eqref{eq:tmp_inf_lb} in terms of the conditional R-D
function~\eqref{eq:def_rd_cond},
\begin{equation}\label{eq:rd_lb}
  n\Info(X; \hat{X}' \mid U) - n\delta_n \geq
  n\RD{d}{D}(\tilde{Q}_{X \mid U}) - n\delta_n,
\end{equation}
and, using~\eqref{eq:rd_lb} and \eqref{eq:tmp_inf_lb}, we obtain the
desired lower bound 
\begin{equation}\label{eq:rd_lb_explicit}
  \E_{Q}[\log(G_n(X^n \mid M))] \geq n\RD{d}{D}(\tilde{Q}_{X \mid U}) - n\delta_n.
\end{equation}
We next return to \eqref{eq:exp_lb} and derive a single-letter expression
for $\D(Q_{X^nY^n} \| P_{X^nY^n})$, where $Q_{X^nY^n}$ is the
$(X^n, Y^n)$-marginal of $Q$, and
\begin{equation}
  \label{eq:amos_clar_300}
 P_{X^nY^n} = P_{XY}^{\times n}. 
\end{equation}
We first express it as
\begin{align}
  \D(Q_{X^nY^n} \| P_{X^nY^n})
  &= \D(Q_{X^nY^n}P_{M \mid Y^n}Q_{U^n \mid X^nY^nM} \| P_{X^nY^n}P_{M \mid
  Y^n}Q_{U^n \mid X^nY^nM}),\label{eq:to_find_d}
\end{align}
and then observe that $Q_{X^nY^n}P_{M \mid Y^n}Q_{U^n \mid X^nY^nM}$ is (a factorization of)
the $(X^n, Y^n, M, U^n)$-marginal of $Q$, which can be
expressed as
\begin{equation}\label{eq:easier_factorization}
  Q_{X^nY^n}P_{M \mid Y^n}Q_{U^n \mid X^nY^nM} = Q_Y^{\times n} \Biggl(
  \prod_{i =
  1}^n Q_{X_i \mid Y_iU_i} \Biggr) P_{M \mid Y^n}Q_{U^n \mid X^nY^nM},
\end{equation}
because, by \eqref{eq:q_full} (or~\eqref{eq:u_i}),
\begin{equation}
  Q_{U^n \mid X^nY^nM} = \prod_{i = 1}^n Q_{U_i \mid X^{i - 1}Y^{i - 1}M}.
\end{equation}

From~\eqref{eq:amos_clar_300}, \eqref{eq:easier_factorization},
and~\eqref{eq:to_find_d}
\begin{align}
  \D(Q_{X^nY^n} \| P_{X^nY^n})
  &= \D(Q_{X^nY^n}P_{M \mid Y^n}Q_{U^n \mid X^nY^nM} \| P_{X^nY^n}P_{M \mid
  Y^n}Q_{U^n \mid X^nY^nM})\\
  &= \D\left(Q_Y^{\times n} \prod_{i = 1}^n Q_{X_i \mid Y_iU_i}
     P_{M \mid
  Y^n}Q_{U^n \mid X^nY^nM} \Bigg{\|} P_{XY}^{\times n}P_{M \mid Y^n}Q_{U^n \mid
  X^nY^nM}\right).\label{eq:to_single_letter}
\end{align}
We now continue the derivation of a single-letter expression for
$\D(Q_{X^nY^n} \| P_{X^nY^n})$ by studying the RHS of
\eqref{eq:to_single_letter}:
\begin{align}
  &\D(Q_{X^nY^n} \| P_{X^nY^n})\nonumber\\
  &= \D\left(Q_Y^{\times n}\prod_{i = 1}^n Q_{X_i \mid Y_iU_i} P_{M \mid
  Y^n}Q_{U^n \mid X^nY^nM} \Bigg{\|} P_{XY}^{\times n}P_{M \mid Y^n}Q_{U^n
  \mid X^nY^nM}\right)\\
  &\stackrel{(a)}{=} \E_Q\left[\log\left(\frac{Q_Y^{\times n}(Y^n)\prod_{i =
  1}^n Q_{X_i \mid Y_iU_i}(X_i \mid Y_i, U_i) P_{M \mid Y^n}(M \mid
  Y^n)Q_{U^n \mid X^nY^nM}(U^n \mid X^n, Y^n, M)}{P_{XY}^{\times n}(X^n,
  Y^n)P_{M \mid Y^n}(M \mid Y^n)Q_{U^n \mid X^nY^nM}(U^n \mid X^n, Y^n,
  M)}\right)\right]\\
  &= \E_Q\left[\log\left(\frac{Q_Y^{\times n}(Y^n)\prod_{i = 1}^n Q_{X_i \mid
  Y_iU_i}(X_i \mid Y_i, U_i)}{P_{XY}^{\times n}(X^n, Y^n)}\right)\right]\\
  &\stackrel{(b)}{=} \sum_{i =
  1}^n\E_{Q_{X_iY_iU_i}}\left[\log\left(\frac{Q_Y(Y_i)Q_{X_i \mid
  Y_iU_i}(X_i \mid Y_i, U_i)}{P_{XY}(X_i, Y_i)}\right)\right]\\
  &= \sum_{i = 1}^n \sum_{(x_i, y_i, u_i) \in \set{X} \times \set{Y} \times
  \set{U}_i} Q_{X_iY_iU_i}(x_i, y_i, u_i)\log\left(\frac{Q_Y(y_i)Q_{X_i \mid
  Y_iU_i}(x_i \mid y_i, u_i)}{P_{XY}(x_i, y_i)}\right)\\
  &\stackrel{(c)}{=} \sum_{i = 1}^n \sum_{(x_i, y_i, u_i) \in \set{X} \times
  \set{Y} \times \set{U}_i} Q_{X_iY_iU_i}(x_i, y_i,
  u_i)\log\left(\frac{Q_{Y_i}(y_i)Q_{X_i \mid Y_iU_i}(x_i \mid y_i,
  u_i)}{P_{XY}(x_i, y_i)}\right)\\
  &= n\sum_{i = 1}^n \sum_{(x_i, y_i, u_i) \in \set{X} \times \set{Y} \times
  \set{U}_i} \frac{1}{n}Q_{X_iY_iU_i}(x_i, y_i,
  u_i)\log\left(\frac{Q_{Y_i}(y_i)Q_{U_i \mid Y_i}(u_i \mid y_i)Q_{X_i \mid
  Y_iU_i}(x_i \mid y_i, u_i)\frac{1}{n}}{P_{XY}(x_i, y_i)Q_{U_i \mid
  Y_i}(u_i \mid y_i)\frac{1}{n}}\right)\\
  &\stackrel{(d)}{=} n\sum_{i = 1}^n \sum_{(x_i, y_i, u_i) \in \set{X} \times
  \set{Y} \times \set{U}_i} \tilde{Q}(x_i, y_i,
  u_i)\log\left(\frac{\tilde{Q}(x_i, y_i, u_i)}{P_{XY}(x_i, y_i)\tilde{Q}_{U
  \mid Y}(u_i \mid y_i)}\right)\\
  &= n\sum_{(x, y, u) \in \set{X} \times \set{Y} \times (\cup_{i = 1}^n
  \set{U}_i)} \tilde{Q}(x, y, u)\log\left(\frac{\tilde{Q}(x, y, u)}{P_{XY}(x,
  y)\tilde{Q}_{U \mid Y}(u \mid y)}\right)\\
  &=n\D(\tilde{Q}_{XYU} \| P_{XY}\tilde{Q}_{U \mid
  Y}),\label{eq:d_single_letter}
\end{align}
where (a) follows from the definition of the relative entropy and the fact
that $Q_Y^{\times n}\prod_{i = 1}^n Q_{X_i \mid Y_iU_i} P_{M \mid
Y^n}Q_{U^n \mid X^nY^nM}$ is (a factorization of) the
$(X^n, Y^n, M, U^n)$-marginal of $Q$; in (b) we have used that for
nonnegative $x$ and $y$, $\log(xy) = \log(x) + \log(y)$, and we used
$Q_{X_iY_iU_i}$ to denote the $(X_i, Y_i, U_i)$-marginal of $Q$;
 (c) holds because under~$Q$, $Y^n \sim \IID Q_Y$;
and in (d) we have identified
$\frac{1}{n}Q_{X_iY_iU_i}$ as the $(X, Y, U)$-marginal of $\tilde{Q}$.

We next show that, $\Info_{\tilde{Q}}(Y; U)$---the mutual information
between $Y$ and $U$ under $\tilde{Q}$---is upper-bounded by $R$. To
that end first observe that by definition of $\tilde{Q}$ (in
\eqref{eq:amos_clar_500} and~\eqref{eq:q_mix}) we can express
$\Info_{\tilde{Q}}(Y; U)$ as
\begin{equation}\label{eq:block_1a}
  \Info_{\tilde{Q}}(Y; U) 
  =
  \frac{1}{n}\sum_{i = 1}^n
  \Bigl(\Entr_{Q}(Y_i) - \Entr_{Q}(Y_i \mid U_i)\Big).
\end{equation}
So continuing from the RHS of \eqref{eq:block_1a}, with all information-theoretic
quantities implicitly evaluated w.r.t.~$Q$:
\begin{align}
  \frac{1}{n}\sum_{i = 1}^n \Bigl(\Entr(Y_i) - \Entr(Y_i \mid
  U_i)\Big) &= \frac{1}{n}\sum_{i = 1}^n \Bigl(\Entr(Y_i) - \Entr(Y_i
  \mid X^{i - 1}, Y^{i - 1}, M)\Big)\\
  &\stackrel{(a)}{=} \frac{1}{n}\sum_{i = 1}^n \Bigl(\Entr(Y_i)
  - \Entr(Y_i \mid Y^{i - 1}, M)\Big)\\
  &\stackrel{(b)}{=} \frac{1}{n}\sum_{i = 1}^n \Bigl(\Entr(Y_i
  \mid Y^{i - 1}) - \Entr(Y_i \mid Y^{i - 1}, M)\Big)\\
  &= \frac{1}{n}\sum_{i = 1}^n \Info(Y_i; M \mid Y^{i - 1})\\
  &= \frac{1}{n}\Info(Y^n; M)\\
  &\leq \frac{1}{n}\Entr(M)\label{eq:block_1b}\\
  &\stackrel{(c)}{\leq} R,\label{eq:info_ub}
\end{align}
where (a) holds because, under $Q$,
$X^{i - 1} \to (Y^{i - 1}, M) \to Y_i$ \eqref{eq:amos_clar700Markov};
(b) holds because $Y^n$ is IID under $Q$; and (c)~holds because $M$
can assume at most $2^{nR}$ distinct values.

We now use \eqref{eq:exp_lb}, \eqref{eq:rd_lb},
\eqref{eq:d_single_letter}, and \eqref{eq:info_ub} to derive the converse
part of Theorem \ref{thm:main} as stated in
\eqref{eq:target_lb}. Starting with~\eqref{eq:exp_lb}, we
use~\eqref{eq:rd_lb} and~\eqref{eq:d_single_letter} to obtain 
\begin{equation}
  \E_{P_{X^nY^n}}[G_n(X^n \mid f_n(Y^n))^\rho] \geq 2^{n(\rho\RD{d}{D}(\tilde{Q}_{X \mid U}) - \D(\tilde{Q}_{XYU} \|
  P_{XY}\tilde{Q}_{U \mid Y}) - \delta_n)},\label{eq:tmp_exp_lb}
\end{equation}
where the PMF $\tilde{Q}$ on the RHS of \eqref{eq:tmp_exp_lb} is defined in
\eqref{eq:q_mix}. Taking the logarithm and dividing by $n$ on both sides,
\begin{equation}
  \frac{1}{n}\log(\E[G_n(X^n \mid f_n(Y^n)^\rho]) \geq \rho\RD{d}{D}(\tilde{Q}_{X \mid U}) - \D(\tilde{Q}_{XYU} \|
  P_{XY}\tilde{Q}_{U \mid Y}) - \delta_n.\label{eq:intermediate_lb}
\end{equation}
Since the choice of $Q_Y$ and $\{Q_{X_i \mid Y_iU_i}\}_{i = 1}^n$ in
\eqref{eq:q_full} is arbitrary, so is that of $\tilde{Q}_Y$ and
$\tilde{Q}_{X \mid YU}$ in the $(X, Y, U)$-marginal
$\tilde{Q}_{XYU} = \tilde{Q}_Y\tilde{Q}_{U \mid Y}\tilde{Q}_{X \mid
  YU}$ of~$\tilde{Q}$~\eqref{eq:q_mix}. We are therefore at liberty to
choose those so as to obtain the tightest bound. Things are different
with regard to $\tilde{Q}_{U \mid Y}$, because it is influenced by the
helper~$f_n$, and we must ensure that the bound is valid for all
helpters.
Ostensibly, we should therefore consider the choice of $\tilde{Q}_{U \mid Y}$ that
yields the loosest bound. However, $\tilde{Q}_{U \mid Y}$ cannot be arbitrary:
%
%
irrespective of our choice of $\tilde{Q}_Y$, the mutual information
$I_{\tilde{Q}}(U; Y)$ must be upper bounded by $R$
\eqref{eq:info_ub}. 

These considerations allow to infer form~\eqref{eq:intermediate_lb} that
\begin{align}
  &\frac{1}{n}\log(\E[G_n(X^n \mid f_n(Y^n)^\rho])\nonumber\\
  &\quad\quad \geq \sup_{\tilde{Q}_Y} \inf_{Q_{\tilde{U} \mid Y}\colon
  \Info(\tilde{Q}_{Y; U}) \leq R} \sup_{\tilde{Q}_{X \mid YU}}
  \Big(\rho\RD{d}{D}(\tilde{Q}_{X \mid U}) - \D(\tilde{Q}_{XYU} \|
  P_{XY}\tilde{Q}_{U \mid Y})\Big) - \delta_n,
\end{align}
which, upon taking $n$ to infinity, yields \eqref{eq:target_lb}.
\end{proof}

\appendices

\section{}\label{app:a}
\begin{lemma}\label{lmm:entropy_lb}
Let $X$ be a chance variable taking values in the finite set $\set{X}$
according to some PMF $P$, and let $f$ be
a bijection from $\set{X}$ to $[1:|\set{X}|]$. Then, for
$X \sim P$,
\begin{equation}\label{eq:entr_lb}
  \E[\log(f(X))] \geq \Entr(X) - \log(\ln(|\set{X}|) + 3/2).
\end{equation}
\end{lemma}

\begin{proof}
Outcomes of zero probability contribute neither to the LHS nor to the
RHS of \eqref{eq:entr_lb}, and we therefore assume w.l.g. that
$P(x) > 0$ for every $x \in \set{X}$. We then have
\begin{align}
  \E[\log(f(X))] &= \sum_{x \in \set{X}} P(x)\log(f(x))\\
  &= \sum_{x \in \set{X}} P(x)\log\left(\frac{f(x)P(x)}{P(x)}\right)\\
  &= \Entr(X) + \sum_x P(x)\log(f(x)P(x))\\
  &= \Entr(X) - \sum_x P(x)\log\left(\frac{1}{f(x)P(x)}\right)\\
  &\stackrel{(a)}{\geq} \Entr(X) - \log\left(\sum_x
  \frac{1}{f(x)}\right)\\
  &\stackrel{(b)}{=} \Entr(X) - \log\left(\sum_{i = 1}^{|\set{X}|}
  \frac{1}{i}\right)\\
  &\stackrel{(c)}{\geq} \Entr(X) - \log(\ln(|\set{X}|) + 3/2),
\end{align}
where (a) follows from Jensen's inequality; (b) holds because $f$
maps onto $[1:|\set{X}|]$; and (c) holds because $\sum_{i = 1}^n
1/i$ is upper-bounded by $\ln(n) + 3/2$.
\end{proof}

\begin{corollary}\label{cor:entropy_lb}
Let $\set{X}$ be a finite set, and $f$ a bijection from
$\set{X}^n$ to $[1:|\set{X}|^n]$. Then, for any chance variable
$X^n$ on~$\set{X}^n$,
\begin{equation}
  \E[\log(f(X^n))] \geq \Entr(X^n) - n\delta_n,
\end{equation}
where $\delta_n = \delta_n(|\set{X}|)$ and for every fixed $|\set{X}|$,
\begin{equation}
  \delta_n \downarrow 0.
\end{equation}
\end{corollary}

\begin{proof}
The corollary follows from Lemma \ref{lmm:entropy_lb} and the
fact that when $|\set{X}|$ is fixed,
\begin{equation}
  \lim_{n \to \infty} \frac{\log(\ln(|\set{X}^n|) + 3/2)}{n} = 0.
\end{equation}
\end{proof}

\section{}\label{app:card}
We prove that restricting $U$ to take values in a set of cardinality
$|\set{Y}| + 1$ does not alter \eqref{eq:target_exponent}.  To that
end, we first express the objective function
in~\eqref{eq:target_exponent}
as an expectation over $U$ of a quantity
$\Psi(Q_{Y|U=u}, Q_{X|YU=u})$ that depends explicitly on $Q_{Y|U=u}$,
$Q_{X|YU=u}$ and implicitly on the given joint PMF $P_{XY}$ and the PMF $Q_{Y}$
(which is determined in the outer maximization). Specifically,
\begin{subequations}
  \label{block:hot200}
\begin{IEEEeqnarray}{rCl}
  \rho\RD{d}{D}(Q_{X \mid U}) - \D(Q_{XYU} \| P_{XY}Q_{U \mid Y}) & =
  & \sum_{u \in \set{U}} Q_{U}(u) \, \Psi(Q_{Y \mid U=u}, Q_{X \mid YU=u}), \label{eq:hot200a}
\end{IEEEeqnarray}
with
\begin{IEEEeqnarray}{rCl}
  \Psi(Q_{Y \mid U=u}, Q_{X \mid YU=u}) & = & \rho\RD{d}{D}(Q_{X \mid U=u}) + 
        H(Q_{Y}) - H(Q_{Y \mid U=u}) - D(Q_{Y \mid
        U = u}Q_{X \mid YU = u} \| P_{XY}) \label{eq:hot200b}
      \IEEEeqnarraynumspace
\end{IEEEeqnarray}
\end{subequations}
where $\RD{d}{D}(Q_{X \mid U=u})$ is determined by $Q_{Y|U=u}$ and
$Q_{X|YU=u}$ via the relation
\begin{IEEEeqnarray}{rCl}
  Q_{X \mid U}(x \mid u) & = & \sum_{y \in \set{Y}} Q_{Y \mid U}(y \mid u) \, Q_{X \mid YU}(x \mid y,u).
\end{IEEEeqnarray}
Indeed, \eqref{block:hot200} follow from
\begin{align}
  \D(Q_{XYU} \| P_{XY}Q_{U \mid Y}) &= \E_{Q_{XYU}} \!\!\Bigg[
  \log \left(\frac{Q_{XY \mid U}(X, Y \mid U) \, Q_{U}(U)}
    {P_{XY}(X, Y) \,Q_{U \mid Y}(U \mid Y)} \right)\Bigg]\\
  &= -H(Q_{U}) + H(Q_{U \mid Y}) + \E_{Q_{XYU}} \!\!\Bigg[
\log\left(\frac{Q_{XY \mid U}(X, Y \mid U)}{P_{XY}(X, Y)} \right) \Bigg]\\
  &= -H(Q_{Y}) + H(Q_{Y \mid U}) + \E_{Q_{XYU}} \!\!\Bigg[
\log\left(\frac{Q_{XY \mid U}(X, Y \mid U)}{P_{XY}(X, Y)} \right) \Bigg]\\
  &= -H(Q_{Y}) + H(Q_{Y \mid U}) + \E_{Q_{XYU}} \!\!\Bigg[
\log\left(\frac{Q_{Y \mid U}(Y \mid U) \, Q_{X|YU}(X \mid Y, U)}{P_{XY}(X, Y)}
\right) \Bigg]\\
  &= - \sum_{u \in \set{U}} Q_{U}(u) \Bigg( H(Q_{Y}) - H(Q_{Y \mid U=u})\nonumber\\
  &\quad\quad - \sum_{(x,y) \in \set{X} \times \set{Y}}
Q_{Y \mid U}(y \mid u) \, Q_{X \mid YU}(x \mid y,u) \log\left(\frac{Q_{Y \mid U}(y \mid u) \,
    Q_{X \mid YU}(x \mid y, u)}{P_{XY}(x, y)}\right)
\Bigg)\\
  &= - \sum_{u \in \set{U}} Q_{U}(u) \Bigg( H(Q_{Y}) - H(Q_{Y \mid U=u}) - D(Q_{Y \mid
        U = u}Q_{X \mid YU = u} \| P_{XY})
\Bigg).
\end{align}

The representation~\eqref{block:hot200} shows that the inner
maximization in~\eqref{eq:target_exponent} can be performed separately
for every $u$. Defining
\begin{IEEEeqnarray}{rCl}
  \Psi^{\ast}(Q_{Y \mid U=u}) = \max_{Q_{X \mid YU=u}} \Psi(Q_{Y \mid U=u}, Q_{X \mid YU=u})
\end{IEEEeqnarray}
we can express~\eqref{eq:target_exponent} as
\begin{IEEEeqnarray}{rCl}
\sup_{Q_Y} \inf_{Q_{U \mid Y}: \Info(Q_{Y;U}) \leq R} \sum_{u \in \set{U}} Q_{U}(u)
\,   \Psi^{\ast}(Q_{Y \mid U=u}).
\end{IEEEeqnarray}
We next view the inner minimization above as being over all pairs $(Q_{U},Q_{Y \mid U})$
with the objective function being 
\begin{IEEEeqnarray}{rCl}
  \sum_{u \in \set{U}} Q_{U}(u)
\,   \Psi^{\ast}(Q_{Y \mid U=u});
\end{IEEEeqnarray}
with the constraint on the $Y$-marginal
\begin{IEEEeqnarray}{rCl}
  \sum_{u \in \set{U}} Q_{U}(u) \, Q_{Y \mid U}(y \mid u) & = & Q_{Y}(y), \qquad \forall y
  \in \set{Y};
\end{IEEEeqnarray}
and the constraint on the mutual information 
\begin{IEEEeqnarray}{rCl}
  \sum_{u \in \set{U}} Q_{U}(u) \, H(Q_{Y \mid U=u}) & \geq & H(Q_{Y}) - R.
\end{IEEEeqnarray}
Since the objective function and constraints are linear in $Q_{U}$, it
follows from Carath\'{e}odory's theorem (for connected sets) that the
cardinality of $\set{U}$ can be restricted to $|\set{Y}| + 1$.

\ifCLASSOPTIONcaptionsoff
  \newpage
\fi


\bibliographystyle{IEEEtran}
\bibliography{./bibliography.bib}




\end{document}